\documentclass[12pt,letterpaper]{article}
\usepackage{xcolor}
\usepackage{fullpage}
\usepackage[utf8]{inputenc}
\usepackage[T1]{fontenc}
\usepackage{amsthm}
\usepackage{amsmath}
\usepackage{amsfonts}
\usepackage{amssymb}
\usepackage{graphicx}
\usepackage{bm}
\usepackage{hyperref}
\usepackage{cleveref}
\usepackage{authblk}

\newcommand{\Hilb}[1][]{\ensuremath{\mathcal{H}_{#1}}}

\newcommand{\Ket}[1]{\ensuremath{\left \vert #1 \right \rangle}}
\newcommand{\Bra}[1]{\ensuremath{\left \langle #1 \right \vert}}
\newcommand{\BraKet}[2]{\ensuremath{\left \langle #1 \middle \vert #2 \right \rangle}}
\newcommand{\SBraKet}[1]{\ensuremath{\BraKet{#1}{#1}}}
\newcommand{\Outer}[2]{\ensuremath{\left \vert #1 \middle \rangle \middle \langle #2 \right \vert}}
\newcommand{\Proj}[1]{\ensuremath{\Outer{#1}{#1}}}
\newcommand{\Sand}[3]{\ensuremath{\left \langle #1 \middle \vert #2 \middle \vert #3 \right \rangle}}

\newcommand{\Abs}[1]{\ensuremath{\left \vert #1 \right \vert}}

\newcommand{\Expect}[1]{\ensuremath{\left \langle #1 \right \rangle}}
\newcommand{\QProb}[2]{\ensuremath{\Abs{\BraKet{#1}{#2}}^2}}

\newcommand{\Tr}[2][]{\ensuremath{\mathrm{Tr}_{#1} \left ( #2 \right )}}

\theoremstyle{definition}
\newtheorem{definition}{Definition}[section]

\theoremstyle{plain}
\newtheorem{theorem}[definition]{Theorem}
\newtheorem{proposition}[definition]{Proposition}
\newtheorem{corollary}[definition]{Corollary}
\theoremstyle{remark}
\newtheorem{remark}[definition]{Remark}

\title{Uncertainty from the Aharonov-Vaidman Identity}
\author{Matthew S. Leifer}
\affil{Institute for Quantum Studies and Schmid College of Science and Technology \authorcr Chapman University, One University Drive, Orange, CA 92866, USA}

\begin{document}

\maketitle

\begin{abstract}
	In this article, I show how the Aharonov-Vaidman identity $A\Ket{\psi} = \Expect{A}\Ket{\psi} + \Delta A \Ket{\psi^{\perp}_A}$ can be used to prove relations between the standard deviations of observables in quantum mechanics.  In particular, I review how it leads to a more direct and less abstract proof of the Robertson uncertainty relation $\Delta A \Delta B \geq \frac{1}{2} \Abs{\Expect{[A,B]}}$ than the textbook proof.  I discuss the relationship between these two proofs and show how the Cauchy-Schwarz inequality can be derived from the Aharonov-Vaidman identity. I give Aharonov-Vaidman based proofs of the Maccone-Pati uncertainty relations and I show how the Aharonov-Vaidman identity can be used to handle propagation of uncertainty in quantum mechanics.  Finally, I show how the Aharonov-Vaidman identity can be extended to mixed states and discuss how to generalize the results to the mixed case.
\end{abstract}

\section{Introduction}

Let $A$ be a Hermitian operator on a Hilbert space $\Hilb$.  Then, for any (not necessarily normalized) vector $\Ket{\psi}\in\Hilb$,
\begin{equation}
	\label{eq:AharonovVaidman}
	A\Ket{\psi} = \Expect{A}\Ket{\psi} + \Delta A \Ket{\psi^{\perp}_A},
\end{equation}
where $\Expect{A} = \Sand{\psi}{A}{\psi}/\SBraKet{\psi}$ is the expectation value of $A$, $\Delta A = \sqrt{\Expect{A^2} - \Expect{A}^2}$ is its standard deviation, and $\Ket{\psi^{\perp}_A}$ is a vector that is orthogonal to $\Ket{\psi}$, has equal norm $\SBraKet{\psi^{\perp}_A} = \SBraKet{\psi}$, and depends on the operator $A$.  

\Cref{eq:AharonovVaidman} is the \emph{Aharonov-Vaidman Identity}, which first appeared in \cite{Aharonov1990}.  Yakir Aharonov has stated that he ``[does not] understand why it doesn't appear in every quantum book''  \cite{Aharonov}.  The main purpose of this article is to explain why it should appear in undergraduate quantum mechanics textbooks\footnote{Other demonstrations of the usefulness of the Aharonov-Vaidman identity include its use in the proof that, for any state $\Ket{\psi}$ and any observable $A$, $\Ket{\psi}^{\otimes n}$ is an approximate eigenstate of the observable $\bar{A} = \frac{1}{n} \sum_{j=1}^n A_j$ for large $n$, where $A_j$ refers to $A$ acting on the $j^{\mathrm{th}}$ subsystem \cite{Aharonov1990}, and its use in deriving the minimum time required for evolution to an orthogonal quantum state \cite{Vaidman1992}.}.  

The uncertainty relation that is proved most often in quantum mechanics classes and textbooks is the Robertson relation \cite{Robertson1929}:
\begin{equation}
	\label{eq:Robertson}
	\Delta A \Delta B \geq \frac{1}{2} \Abs{\Expect{\left [ A, B\right ]}}, 
\end{equation}
where $[A,B] = AB-BA$ is the commutator.

As pointed out by Schr{\"o}dinger \cite{Schroedinger1930}, the Robertson relation can be extended to
\begin{equation}
	\label{eq:Schrodinger}
	\left ( \Delta A \right )^2 \left ( \Delta B \right )^2 \geq \Abs{\frac{1}{2}  \Expect{\left \{ A,B\right \}} -\Expect{A}\Expect{B}}^2 + \Abs{\frac{1}{2}\Expect{[A,B]}}^2,
\end{equation}
where $\{A,B\} = AB + BA$ is the anti-commutator.

Although not often emphasized in quantum mechanics classes, the Schr{\"o}dinger relation is not harder to prove than the Robertson relation.  In fact, the standard textbook proof of the Robertson relation effectively proves the Schr{\"o}dinger relation and then throws away the anti-commutator term.

The proof almost universally adopted in textbooks is based on the Cauchy-Schwarz inequality.  While this proof is elementary for those familiar with the mathematics of Hilbert spaces, it can be daunting for undergraduate physics students, who are likely encountering Hilbert spaces for the first time along with quantum mechanics.  

In this article, I will review more direct proofs of \cref{eq:Robertson} and \cref{eq:Schrodinger} from the Aharonov-Vaidman identity that only make use of basic properties of complex numbers and inner products.  These proofs previously appeared in \cite{Goldenberg1996} and the proof of the Robertson relation is also problem 3.10 in Aharonov and Rohrlich's book ``Quantum Paradoxes'' \cite{Aharonov2005}.  The proof of the Aharonov-Vaidman identity itself is uses similar ideas to one of the standard proofs of the Cauchy-Schwarz identity, but is perhaps more memorable to undergraduate physics students because it uses concepts that have a physical meaning, i.e.\ expectation values and standard deviations.  The proof of the Robertson and Schr{\"o}dinger relations so obtained is not independent of the standard Cauchy-Schwarz based proof.  I shall discuss their relationship and show that the Cauchy-Schwarz inequality can itself be derived from \cref{eq:AharonovVaidman}.  The main virtue of using the Aharonov-Vaidman based proof of the uncertainty relation is that it is more direct and involves fewer abstractions.  

To be clear, I am not against using or teaching the Cauchy-Schwarz inequality.  It has been called ``one of the most widely used and important inequalities in all of mathematics'' \cite{Steele2004}.  In fact, the Aharonov-Vaidman based proof still uses one instance of the Cauchy-Schwarz inequality, namely that if $\Ket{\psi}$ and $\Ket{\phi}$ are unit vectors then $\Abs{\BraKet{\phi}{\psi}} \leq 1$.  But this is easily motivated by the idea that $\BraKet{\phi}{\psi}$ is a generalization of the cosine of an angle, and it is used in a more direct way than in the standard proof.  Students of quantum mechanics also need to know the Cauchy-Schwarz inequality to prove that the Born rule always yields well-defined probabilities.  Physics students should learn the Cauchy-Schwarz inequality.  I just think it should be used in a less abstract way where possible.

Besides the Robertson and Schr{\"o}dinger relations, many other uncertainty relations are known.  Indeed, since uncertainty relations have found applications in quantum information science \cite{Fuchs1996, Hofmann2003, Guehne2004, Koashi2006, Berta2010, Majumdar2016, YungerHalpern2019} and quantum foundations \cite{Oppenheim2010, Catani2022}, proving new ones has become something of a sport.  The two most common classes of uncertainty relations are those based on entropy \cite{Coles2017} and those based on standard deviations \cite{Robertson1929, Schroedinger1930, Pati2007}.  Many of the standard deviation based relations can be derived from the Aharonov-Vaidman relation.  I include a proof of the Maccone-Pati uncertainty relations \cite{Maccone2015} to illustrate this.  While these are not the most recent or tightest known uncertainty relations, I include them because they have a simple and elegant Aharonov-Vaidman based proof.  For more recent work on standard deviation uncertainty relations, see \cite{Bannur2015, Li2015, Yao2015, Abbott2016, Chen2016, Qin2016, Song2016, Xiao2016, Mondal2017, Song2017, Zhang2017, Zheng2017, Dodonov2018, Guise2018, Busch2019, Giorda2019, Zheng2020, Li2021, Zhang2021, Chiew2022, Xiao2022}.

Another place where relationships between standard deviations are important is in the propagation of uncertainty.  In classical statistics, if random variables $X_1, X_2,\cdots,X_n$ have standard deviations $\Delta X_1, \Delta X_2, \cdots \Delta X_n$ then a function of them $f(X_1,X_2,\cdots,X_n)$ has standard deviation $\Delta f$ that is a function of $\Delta X_1, \Delta X_2, \cdots \Delta X_n$ (and their correlations if the variables are not independent).  Formulas for the propagation of uncertainty tell us how to compute this function, and are commonly used to estimate experimental errors.  In quantum mechanics, similar formulas can be derived relating the standard deviations of observables.  They differ from their classical counterparts due to the fact that quantum observables do not commute, but provided this is taken care of they can be derived by the same methods as in the classical case.  However, they can alternatively be derived from the Aharonov-Vaidman identity, as I shall explain.

Although the Aharonov-Vaidman identity is usually discussed for pure quantum states, it can be extended to mixed states, either by use of purification or an equivalent concept called an \emph{amplitude operator}.  Relations between standard deviations can be extended to mixed states, but obtaining tight bounds is sometimes more difficult than in the pure case due to the need to optimize over all purifications or amplitude operators that can represent a given mixed state.

The remainder of this article is structured as follows.  \Cref{AVProof} gives the proof of the Aharonov-Vaidman identity and a corollary that is useful for understanding the equality conditions in uncertainty relations.  \Cref{Robertson} presents the proof of the Robertson and Schr{\"o}dinger relations based on the Aharonov-Vaidman identity.  \Cref{Cauchy} explains the relationship with the standard textbook proof of the Robertson relation and explains how the Cauchy-Schwarz inequality can be derived from the Aharonov-Vaidman identity.  \Cref{Pedagogy} comments on the effective teaching of the Robertson uncertainty relations via the Aharonov-Vaidman identity.  \Cref{Other} presents Aharonov-Vaidman based proofs of the Maccone-Pati uncertainty relations.  \Cref{Prop} describes how to use the Aharonov-Vaidman identity to derive formulas for the propagation of quantum uncertainty.  \Cref{Mixed} explains how to generalize the Aharonov-Vaidman relation to mixed states using amplitude operators.  (The relationship between amplitude operators and purifications is discussed in \cref{App:Purify}.)  Finally, \cref{Conc} presents the summary and conclusions.  

I intend this article  to be pedagogical and self-contained, so as to be accessible to undergraduate students and anyone teaching introductory quantum mechanics.

\section{Proof of the Aharonov Vaidman Identity}

\label{AVProof}

Sometimes, it is useful to generalize the Aharonov-Vaidman identity to non-Hermitian operators, so we prove the more general version here.

\begin{proposition}[The Aharonov-Vaidman Identity]
	\label{Prop:AV}
	Let $A$ be a linear operator on a Hilbert space $\Hilb$ and let $\Ket{\psi}$ be a (not necessarily normalized) vector in $\mathcal{H}$.  Then,
	\begin{equation}
		\label{eg:GenAV}
		A \Ket{\psi} = \Expect{A} \Ket{\psi} + \Delta A \Ket{\psi^{\perp}_A},
	\end{equation}
	where $\Expect{A} = \Sand{\psi}{A}{\psi}/\SBraKet{\psi}$, $\Delta A = \sqrt{\Expect{A^{\dagger}A} - \Abs{\Expect{A}}^2}$, and $\Ket{\psi^{\perp}_A}$ is a vector orthogonal to $\Ket{\psi}$ that depends on both $\Ket{\psi}$ and $A$ and satisfies $\SBraKet{\psi^{\perp}_A} = \SBraKet{\psi}$.
\end{proposition}

Note that, if $A$ is Hermitian, then this reduces to \cref{eq:AharonovVaidman}, where $\Expect{A}$ and $\Delta A$ are the expectation value and standard deviation.  In general, $\Expect{A}$ is a complex number, but $\Delta A$ is always real and non-negative.

For most of what we need to do, it is sufficient to consider the case where $\Ket{\psi}$ is a unit vector, in which case $\Ket{\psi^{\perp}_A}$ is also a unit vector.  The exception is the proof of the Cauchy-Schwarz inequality (\cref{prop:CS} in \cref{Cauchy}), which uses the identity with an unnormalized vector.

\begin{proof}
	Given a vector $\Ket{\psi} \in \Hilb$, any other vector $\Ket{\phi}\in\Hilb$ can be written as $\Ket{\phi} = \alpha \Ket{\psi} + \beta \Ket{\psi^{\perp}}$, where $\alpha$ and $\beta$ are complex coefficients and $\Ket{\psi^{\perp}}$ is some vector that is orthogonal to $\Ket{\psi}$.  By an appropriate rescaling of $\beta$, we can ensure that $\SBraKet{\psi^{\perp}} = \SBraKet{\psi}$.  Applying this to $\Ket{\phi} = A\Ket{\psi}$ gives
	\begin{equation}
		\label{eq:Decomp}
		A\Ket{\psi} = \alpha \Ket{\psi} + \beta \Ket{\psi^{\perp}}.
	\end{equation}
	To determine $\alpha$, take the inner product of \cref{eq:Decomp} with $\Ket{\psi}$, which gives
	\begin{equation}
		\Sand{\psi}{A}{\psi} = \alpha\SBraKet{\psi}.
	\end{equation}
	Rearranging this gives $\alpha=\Expect{A}$.
	
	To determine $\beta$, substitute $\alpha = \Expect{A}$ into \cref{eq:Decomp} and take the inner product of $A\Ket{\psi}$ with itself to obtain
	\begin{align*}
		\Bra{\psi}A^{\dagger}A\Ket{\psi} & = \Abs{\Expect{A}}^2 \SBraKet{\psi} + |\beta|^2 \SBraKet{\psi^{\perp}} \\
		& =  \Abs{\Expect{A}}^2 \SBraKet{\psi} + |\beta|^2 \SBraKet{\psi}, 
	\end{align*}
	where we have used $\SBraKet{\psi^{\perp}} = \SBraKet{\psi}$.

	Rearranging and using $\Expect{A^{\dagger}A} = \Sand{\psi}{A^{\dagger}A}{\psi}/\SBraKet{\psi}$ gives
	\begin{equation}
		|\beta|^2 = \Expect{A^{\dagger}A} - \Abs{\Expect{A}}^2 = (\Delta A)^2.
	\end{equation}
	This means that $\beta = (\Delta A) e^{i\theta}$ for some phase angle $\theta$.  If we define $\Ket{\psi^{\perp}_A} = e^{i\theta}\Ket{\psi^{\perp}}$ then $\Ket{\psi^{\perp}_A}$ is still orthogonal to $\Ket{\psi}$, its norm is unchanged, and we have \cref{eg:GenAV}.
\end{proof}

The following corollary is useful for finding the conditions for equality in uncertainty relations.

\begin{corollary}
	\label{cor:cor}
	In general, for two operators $A$ and $B$, and for a unit vector $\Ket{\psi}$,
	\begin{equation}
	\label{eq:covariance}
	\BraKet{\psi^{\perp}_A}{\psi^{\perp}_B} = \frac{\Expect{A^{\dagger}B} - \Expect{A}^*\Expect{B}}{\Delta A\Delta B}.
	\end{equation}
\end{corollary}

\begin{proof}
	From \cref{Prop:AV}, we have
	\begin{align}
		A\Ket{\psi} & = \Expect{A}\Ket{\psi} + \Delta A \Ket{\psi^{\perp}_A}, \\
		B\Ket{\psi} & = \Expect{B}\Ket{\psi} + \Delta B \Ket{\psi^{\perp}_B}.
	\end{align}
	Taking the inner product of these gives
	\begin{equation}
		\Bra{\psi}A^{\dagger}B\Ket{\psi} = \Expect{A}^*\Expect{B} + \Delta A \Delta B \BraKet{\psi^{\perp}_A}{\psi^{\perp}_B},
	\end{equation}
	Rearranging gives the desired result.
\end{proof}

Note that, if $A$ and $B$ are Hermitian then we have
\begin{equation}
\label{eq:correlation}
\BraKet{\psi^{\perp}_A}{\psi^{\perp}_B} = \frac{\Expect{AB} - \Expect{A}\Expect{B}}{\Delta A\Delta B}.
\end{equation}
If it is also the case that $[A,B] = 0$ then \cref{eq:correlation} is the \emph{correlation}, denoted $\text{corr}_{A,B}$, that would be obtained from a joint measurement of $A$ and $B$.  The correlation is a well-known statistical measure of how two random variables are related to one another.  \Cref{eq:correlation} is a formal generalization of the correlation, so we will also denote it $\text{corr}_{A,B}$. However, if $A$ and $B$ do not commute then $\text{corr}_{A,B}$ is generally a complex number, there is no joint measurement of $A$ and $B$ of which $\text{corr}_{A,B}$ could be the correlation, and $AB$ is not an observable.

The real and imaginary parts of $\text{corr}_{A,B}$ are
\begin{align}
	\text{Re} \left ( \text{corr}_{A,B} \right ) & = \frac{1}{2} \left ( \BraKet{\psi^{\perp}_A}{\psi^{\perp}_B} + \BraKet{\psi^{\perp}_B}{\psi^{\perp}_A} \right ) = \frac{\frac{1}{2}\Expect{\left \{ A, B \right \}} - \Expect{A}\Expect{B}}{\Delta A\Delta B} \label{eq:Rcorr} \\
	\text{Im} \left ( \text{corr}_{A,B} \right ) & = \frac{1}{2i} \left ( \BraKet{\psi^{\perp}_A}{\psi^{\perp}_B} - \BraKet{\psi^{\perp}_B}{\psi^{\perp}_A}\right ) = \frac{\Expect{\left [ A, B \right ]}}{2i\Delta A\Delta B}, \label{eq:Imcorr}
\end{align}
The real part is also a formal generalization of the correlation in that it reduces to the classical formula when $A$ and $B$ commute.  We denote it $\text{Rcorr}_{A,B}$.

\section{The Robertson and Schr{\"o}dinger Uncertainty Relations}

\label{Robertson}

We are now in a position to prove the Robertson and Schr{\"o}dinger uncertainty relations.

\begin{proposition}[The Robertson Uncertainty Relation]
	Let $A$ and $B$ be two Hermitian operators on a Hilbert space \Hilb.  Then, for any unit vector $\Ket{\psi} \in \Hilb$
	\begin{equation}
		\label{eq:Robertson2}
		\Delta A \Delta B \geq \frac{1}{2}\Abs{\Expect{[A,B]}}.
	\end{equation}
\end{proposition}

\begin{proof}
	From the Aharonov-Vaidman identity, we have
	\begin{align}
		A\Ket{\psi} & = \Expect{A}\Ket{\psi} + \Delta A \Ket{\psi^{\perp}_A}, \\
		B\Ket{\psi} & = \Expect{B}\Ket{\psi} + \Delta B \Ket{\psi^{\perp}_B}.
	\end{align}
	Taking the inner product of these two equations and its complex conjugate gives
	\begin{align}
		\Sand{\psi}{AB}{\psi} & = \Expect{A}\Expect{B} + \Delta A \Delta B \BraKet{\psi^{\perp}_A}{\psi^{\perp}_B} \label{eq:product1} \\
		\Sand{\psi}{BA}{\psi} & = \Expect{A}\Expect{B} + \Delta A \Delta B \BraKet{\psi^{\perp}_B}{\psi^{\perp}_A}. \label{eq:product2}
	\end{align}
	Subtracting these two equations gives
	\begin{equation}
		\Sand{\psi}{(AB-BA)}{\psi} = \Delta A \Delta B \left (\BraKet{\psi^{\perp}_A}{\psi^{\perp}_B} - \BraKet{\psi^{\perp}_B}{\psi^{\perp}_A} \right ),
	\end{equation}
	or,
	\begin{equation}
		\label{eq:commexpect}
		\Expect{[A,B]} = \Delta A \Delta B \left (\BraKet{\psi^{\perp}_A}{\psi^{\perp}_B} - \BraKet{\psi^{\perp}_B}{\psi^{\perp}_A} \right ),
	\end{equation}	
	Since $\BraKet{\psi^{\perp}_B}{\psi^{\perp}_A}$ is the complex conjugate of $\BraKet{\psi^{\perp}_A}{\psi^{\perp}_B}$, we can rewrite this as
	\begin{equation}
		\Expect{[A,B]} = 2 i \Delta A \Delta B \mathrm{Im} \left ( \BraKet{\psi^{\perp}_A}{\psi^{\perp}_B} \right ). 
	\end{equation}
	Taking the absolute value of both sides and rearranging gives
	\begin{equation}
		\Delta A \Delta B \Abs{\mathrm{Im} \left ( \BraKet{\psi^{\perp}_A}{\psi^{\perp}_B} \right )} = \frac{1}{2} \Abs{\Expect{[A,B]}}.
	\end{equation}
	Because $\Ket{\psi^{\perp}_A}$ and $\Ket{\psi^{\perp}_B}$ are unit vectors, $0 \leq \QProb{\psi^{\perp}_A}{\psi^{\perp}_B} \leq 1$, and hence the absolute value of the imaginary part of $\BraKet{\psi^{\perp}_B}{\psi^{\perp}_A}$ is also bounded between $0$ and $1$.  Hence, we have
	\begin{equation}
		\Delta A \Delta B \geq \frac{1}{2} \Abs{\Expect{[A,B]}}.
	\end{equation}
\end{proof}

The condition for equality in the Robertson relation is $\Abs{\text{Im} \left ( \BraKet{\psi^{\perp}_A}{\psi^{\perp}_B}\right )} = 1$ or, equivalently, $\text{corr}_{A,B} = \pm i$.  States that saturate the inequality are called \emph{(Robertson) intelligent states}.  The condition $\text{corr}_{A,B} = \pm i$ can be used to find intelligent states, although this is not easier than solving for equality in the Robertson relation directly.

\begin{proposition}[The Schr{\"o}dinger Uncertainty Relation]
	Let $A$ and $B$ be two Hermitian operators on a Hilbert space \Hilb.  Then, for any unit vector $\Ket{\psi} \in \Hilb$
	\begin{equation}
	\label{eq:Schrodinger2}
	\left ( \Delta A \right )^2 \left ( \Delta B \right )^2 \geq \Abs{\frac{1}{2}  \Expect{\left \{ A,B\right \}} -\Expect{A}\Expect{B}}^2 + \Abs{\frac{1}{2}\Expect{[A,B]}}^2.
	\end{equation}
\end{proposition}

\begin{proof}
	Taking the sum of \cref{eq:product1} and \cref{eq:product2} gives
	\begin{equation}
		\Expect{\left \{ A,B \right \}} = 2\Expect{A}\Expect{B} + \Delta A \Delta B \left ( \BraKet{\psi^{\perp}_A}{\psi^{\perp}_B} + \BraKet{\psi^{\perp}_B}{\psi^{\perp}_A}\right ),
	\end{equation}
	or,
	\begin{equation}
		\Expect{\left \{ A,B \right \}} - 2\Expect{A}\Expect{B} = \Delta A \Delta B \left ( \BraKet{\psi^{\perp}_A}{\psi^{\perp}_B} + \BraKet{\psi^{\perp}_B}{\psi^{\perp}_A}\right ).
	\end{equation}
	Adding this to \cref{eq:commexpect} gives
	\begin{equation}
		\Expect{\left \{ A,B \right \}} - 2\Expect{A}\Expect{B} + \Expect{[A,B]} = 2 \Delta A \Delta B \BraKet{\psi^{\perp}_A}{\psi^{\perp}_B},
	\end{equation}
	or,
	\begin{equation}
		\label{eq:schromid}
		\Delta A \Delta B \BraKet{\psi^{\perp}_A}{\psi^{\perp}_B} = \frac{1}{2}\Expect{\left \{ A,B \right \}} - \Expect{A}\Expect{B} + \frac{1}{2} \Expect{[A,B]}.
	\end{equation}	
	Now, because $A$ and $B$ are Hermitian, $\{A,B\}$ is Hermitian and $[A,B]$ is anti-Hermitian.  Therefore $\Expect{\{A,B\}}$ is real and $\Expect{[A,B]}$ is imaginary.  Further $\Expect{A}$, $\Expect{B}$, $\Delta A$ and $\Delta B$ are real.  Therefore, taking the modulus squared of \cref{eq:schromid} gives
	\begin{equation}
		(\Delta A)^2 (\Delta B)^2 \QProb{\psi^{\perp}_A}{\psi^{\perp}_B} = \Abs{\frac{1}{2}\Expect{\left \{ A,B \right \}} - \Expect{A}\Expect{B}}^2 + \Abs{\frac{1}{2} \Expect{[A,B]}}^2.		
	\end{equation}
	Finally, because $\Ket{\psi^{\perp}_A}$ and $\Ket{\psi^{\perp}_B}$ are unit vectors, we have $0 \leq \QProb{\psi^{\perp}_A}{\psi^{\perp}_B} \leq 1$, from which the result follows.
\end{proof}

The condition for equality in the Schr{\"o}dinger relation is $\QProb{\psi^{\perp}_A}{\psi^{\perp}_B} = \Abs{\text{corr}_{A,B}}^2 = 1$.  States that saturate the inequality are called \emph{(Schr{\"o}dinger) intelligent states}.  The condition $\Abs{\text{corr}_{A,B}}^2 = 1$ can be used to find intelligent states, although this is not easier than solving for equality in the Schr{\"o}dinger relation directly.

\section{The Textbook Proof and The Cauchy-Schwarz Inequality}

\label{Cauchy}

The textbook proofs of the Robertson and Schr{\"o}dinger uncertainty relations are based on the Cauchy-Schwarz inequality
\begin{equation}
	\QProb{f}{g} \leq \SBraKet{f}\SBraKet{g}.
\end{equation}
Note that the proofs given in \cref{Robertson} also make use of a special case of this inequality: that for unit vectors $\QProb{f}{g} \leq 1$.  This is applied to $\Ket{f} = \Ket{\psi^{\perp}_A}$, $\Ket{g} = \Ket{\psi^{\perp}_B}$.  My aim is not to eliminate any use of the Cauchy-Schwarz inequality, but just to argue that the proof is more memorable if the inequality is applied in a different way than in the standard proof.

In the standard proof, the Cauchy-Schwarz inequality is applied to the two vectors $\Ket{f} = \left (A - \Expect{A} \right )\Ket{\psi}$ and $\Ket{g} = \left (B - \Expect{B} \right )\Ket{\psi}$ to obtain
\begin{equation}
	\Abs{\Sand{\psi}{(A - \Expect{A})(B - \Expect{B})}{\psi}}^2 \leq \Sand{\psi}{(A - \Expect{A})^2}{\psi}\Sand{\psi}{(B- \Expect{B})^2}{\psi}.
\end{equation}
A few lines of messy algebra and cancellations, which I will spare you the details of, yields
\begin{equation}
	\left ( \Delta A \right )^2 \left ( \Delta B \right )^2 \geq \Abs{\frac{1}{2} \Expect{\left \{ A,B \right \}} - \Expect{A}\Expect{B} + \frac{1}{2} \Expect{\left [ A,B \right ]}}^2, 
\end{equation}
from which we can derive the Schr{\"o}dinger and Robertson relations by recognizing the real and imaginary parts of the right hand side.

As physics students do not often see the Cauchy-Schwarz inequality prior to their first course on quantum mechanics, most textbooks include a proof of this as well.  One of the common proofs uses reasoning similar to that which we used to establish the Aharonov-Vaidman identity.  It starts by recognizing that $\Ket{g}$ can be written as
\begin{equation}
	\label{eq:CSDeriv}
	\Ket{g} = \alpha \Ket{f} + \beta \Ket{f^{\perp}},
\end{equation}
where $\Ket{f^{\perp}}$ is a unit vector that is orthogonal to $\Ket{f}$.  To find $\alpha$, take the inner product of this with $\Ket{f}$, which yields $\alpha = \BraKet{f}{g}/\SBraKet{f}$.  Substituting this back into \cref{eq:CSDeriv} and then taking the inner product of $\Ket{g}$ with itself gives  
\begin{equation}
	\SBraKet{g} = \frac{\QProb{f}{g}}{\SBraKet{f}} + \Abs{\beta}^2.
\end{equation}
The Cauchy-Schwarz inequality follows from this by recognizing that $\Abs{\beta}^2$ is real and non-negative.

Summarizing, the standard proof of the Robertson inequality  consists of: proving the Cauchy-Schwarz inequality and then finding convenient vectors to insert into the inequality that will yield terms involving $\Delta A$ and $\Delta B$ after some algebra.  From the Aharonov-Vaidman identity, we can see that the reason the choice $\Ket{f} = \left (A - \Expect{A} \right )\Ket{\psi}$ and $\Ket{g} = \left (B - \Expect{B} \right )\Ket{\psi}$ is guaranteed work is that $\Ket{f} = \Delta A \Ket{\psi^{\perp}_A}$ and $\Ket{g} = \Delta B \Ket{\psi^{\perp}_B}$.

After inserting these choices, one has to multiply out and simplify the expressions in the Cauchy-Schwarz inequality.  This involves recognizing things like $\Expect{A}\Sand{\psi}{A}{\psi} = \Expect{A}^2$ and then canceling several terms.  It is difficult for students to follow the full details of this in a lecture.  In the approach using the Aharonov-Vaidman relation, we already have expressions involving $\Delta A$ and $\Delta B$, so it is easier to see how to get an expression involving $\Delta A \Delta B$.  This expression has fewer terms and there is less cancellation to do.

Although the approach using the Aharonov-Vaidman identity uses the Cauchy-Schwarz inequality in a less convoluted way, it uses similar mathematical ideas.  For vectors $\Ket{f}$ and $\Ket{g}$, we can write $\Ket{g}$ in terms of $\Ket{f}$ and an orthogonal vector, as in the proof of Cauchy-Schwarz, or we can write both vectors in terms of a third vector $\Ket{h}$ as
\begin{align}
	\Ket{f} & = \alpha_1 \Ket{h} + \beta_1 \Ket{h^{\perp}_f} \\
	\Ket{g} & = \alpha_2 \Ket{h} + \beta_2 \Ket{h^{\perp}_g}, 
\end{align}
where $\Ket{h^{\perp}_f}$ and $\Ket{h^{\perp}_g}$ are (generally different) vectors orthogonal to $\Ket{h}$ and $\alpha_1,\beta_1,\alpha_2,\beta_2$ are complex coefficients.  This is what we do in the proof of the Aharonov-Vaidman identity with the choices $\Ket{f} = A\Ket{\psi}$, $\Ket{g} = B\Ket{\psi}$ and $\Ket{h} = \Ket{\psi}$.  The advantage of this approach is that it immediately yields expressions involving the expectation values and standard deviations of the observables, which it is easy to see what to do with in order to get the uncertainty relations.  From this point of view, the standard proof looks like shoehorning something into the Cauchy-Schwarz inequality that will yield standard deviations, and then backtracking to a point more easily obtained from the Aharonov-Vaidman identity.  At the end of the day, both approaches use the same mathematics, but the Aharonov-Vaidman approach does so in a simpler and more direct way.

I would go so far as to say that whenever you are tempted to use the Cauchy-Schwarz inequality to prove a relationship between standard deviations of observables in quantum mechanics, you will have an easier time working from the Aharonov-Vaidman identity (and the special case $\QProb{f}{g} \leq 1$ of the Cauchy-Schwarz inequality for unit vectors) instead.  \Cref{Other} and \Cref{Prop} give more examples of this.

I end this section by showing that you can prove the Cauchy-Schwarz inequality from the Aharonov-Vaidman identity.  I include this not because I think it is the best way to prove the Cauchy-Schwarz inequality, but because finding alternative proofs of the Cauchy-Schwarz inequality is the mathematician's equivalent of the sport of finding new uncertainty relations in quantum mechanics.  It also shows that, in principle, there is nothing that can be proved using the Cauchy-Schwarz inequality that could not be proved using the Aharonov-Vaidman identity.  Of course, outside the context of standard deviations in quantum mechanics, using the Aharonov-Vaidman identity instead of the Cauchy-Schwarz inequality is unlikely to yield a better proof.

\begin{proposition}[Cauchy-Schwarz Inequality]
	\label{prop:CS}
	Let $\Ket{f}$ and $\Ket{g}$ be two vectors in a Hilbert space $\Hilb$.  Then
	\begin{equation}
		\Abs{\BraKet{f}{g}}^2 \leq \SBraKet{f} \SBraKet{g}.
	\end{equation}
\end{proposition}

\begin{proof}
	First note that the inequality trivially holds whenever $\BraKet{f}{g} = 0$ and that $\SBraKet{f} = 0$ implies $\BraKet{f}{g} = 0$.  Therefore, we can assume that both $\BraKet{f}{g} \neq 0$ and $\SBraKet{f} > 0$. 
	
	Let $P = \Proj{g}$.  Note this is not necessarily a projector because $\Ket{g}$ does not have to be normalized, but it is a Hermitian operator.  Applying the Aharonov-Vaidman identity to $P$ and $\Ket{f}$ gives
	\begin{equation}
		P \Ket{f} = \Expect{P} \Ket{f} + \Delta P \Ket{f^{\perp}_{P}},
	\end{equation}
	or equivalently
	\begin{equation}
		\label{eq:AVProj}
		\Ket{g}\BraKet{g}{f} = \frac{\BraKet{f}{g}\BraKet{g}{f}}{\SBraKet{f}} \Ket{f} + \Delta P \Ket{f^{\perp}_P}.
	\end{equation}
	 Taking the inner product with $\Ket{f^{\perp}_P}$ gives
	 \begin{equation}
	 	\BraKet{f^{\perp}_P}{g}\BraKet{g}{f} = \Delta P \SBraKet{f},
	 \end{equation}
       where we used the fact that $\SBraKet{f^{\perp}_P} = \SBraKet{f}$
	 Rearranging and taking the complex conjugate gives
	 \begin{equation}
	 	\label{eq:PerpSubber}
	 	\BraKet{g}{f^{\perp}_P} = \frac{\Delta P\SBraKet{f}}{\BraKet{f}{g}}.
	 \end{equation}
	 
	 Now, taking the inner product of \cref{eq:AVProj} with $\Ket{g}$ gives
	 \begin{equation}
	 	\SBraKet{g}\BraKet{g}{f} = \frac{\BraKet{f}{g}\BraKet{g}{f}}{\BraKet{f}{f}} \BraKet{g}{f} + \Delta P \BraKet{g}{f^{\perp}_P}.
	 \end{equation}
	 Multiplying both sides by $\SBraKet{f}/\BraKet{g}{f}$ gives
	 \begin{equation}
	 	\SBraKet{f}\SBraKet{g} = \BraKet{f}{g}\BraKet{g}{f}  +  \frac{\Delta P \BraKet{g}{f^{\perp}_P} \SBraKet{f}}{\BraKet{g}{f}}.
	 \end{equation}
	 Substituting \cref{eq:PerpSubber} into this gives
	 \begin{equation}
	 	\SBraKet{f}\SBraKet{g} = \BraKet{f}{g}\BraKet{g}{f}  +  \frac{\left ( \Delta P \right )^2  \QProb{f}{f}}{\BraKet{f}{g}\BraKet{g}{f}},
	 \end{equation}
	 or
	 \begin{equation}
		\SBraKet{f}\SBraKet{g} = \QProb{f}{g}  +  \frac{\left ( \Delta P \right )^2  \QProb{f}{f}}{\QProb{f}{g}}.
	\end{equation}
	Now, the terms $\Delta P$, $\SBraKet{f}$ and $\Abs{\BraKet{f}{g}}$ are all real and non-negative.  Hence,
	 \begin{equation}
		\SBraKet{f}\SBraKet{g} \geq \QProb{f}{g}.
	\end{equation}	
\end{proof}

\section{Pedagogical Notes}

\label{Pedagogy}

In order to teach the Robertson uncertainty relation via the Aharonov-Vaidman identity, you first have to establish the Aharonov-Vaidman identity.  For the purposes of proving the Robertson uncertainty relation, it is sufficient to restrict the operator in the identity to be Hermitian and the vector $\Ket{\psi}$ to be a unit vector, as I shall in this section.  

In my experience, not all students immediately understand why, given a unit vector $\Ket{\psi}$, any other unit vector $\Ket{\phi}$ can be written as
\begin{equation}
	\label{eq:Convince}
	\Ket{\phi} = \alpha \Ket{\psi} + \beta \Ket{\psi^{\perp}},
\end{equation}
where $\Ket{\psi^{\perp}}$ is a unit vector orthogonal to $\Ket{\psi}$.  They will probably have seen Gram-Schmidt orthogonalization in a linear algebra class, but may have difficulty using that knowledge here due to the jump to abstract Hilbert spaces and Dirac notation.  To aid intuition, I remark that $\Ket{\psi}$ and $\Ket{\phi}$ span a two-dimensional subspace of $\Hilb$ and show them \cref{fig:Ortho}.
\begin{figure}[!htb]
	\centering
	\includegraphics[width=0.3\linewidth]{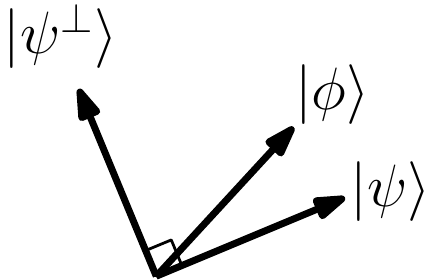}
	\caption{\label{fig:Ortho}Diagram showing that there exists a unit vector $\Ket{\psi^{\perp}}$ such that $\Ket{\psi}$ and $\Ket{\psi^{\perp}}$ form an orthogonal basis for the two dimensional subspace of $\Hilb$ spanned by $\Ket{\psi}$ and $\Ket{\phi}$.}
\end{figure}
By the process of Gram-Schmidt orthogonalization, we can construct an orthornormal basis for this subspace consisting of $\Ket{\psi}$ and
\begin{equation}
	\Ket{\psi^{\perp}} = \frac{1}{\sqrt{1 - \QProb{\phi}{\psi}}} \left ( \Ket{\phi} - \Ket{\psi}\BraKet{\psi}{\phi} \right ),
\end{equation} 
from which we have \cref{eq:Convince} with $\alpha = \BraKet{\psi}{\phi}$ and $\beta = \sqrt{1 - \QProb{\phi}{\psi}}$.

In my quantum mechanics classes, I set students in-class activities that involve things like deriving important equations or making order of magnitude estimates.  These take about 5-10 minutes each and are done in pairs.  I usually do two or three such activities per class.  I believe this increases active engagement and retention of the main principles.  I try to reduce the number of long derivations that I do myself on the board because I think they cause confusion about what the most important equations are and the derivations are rarely remembered by the students.  However, I also do not want to set the students a long and complicated derivation to do themselves in class, so I try to find shorter derivations that they can do with guidance instead.  The proof of the Robertson relation from the Aharonov-Vaidman relation is better suited to this approach than the standard proof.

After establishing \cref{eq:Convince}, I set students the following activity.

\vspace{1ex}
\noindent{\textbf{In Class Activity}}

\vspace{1ex}
\emph{Given that $A\Ket{\psi} = \alpha \Ket{\psi} + \beta \Ket{\psi^{\perp}}$, find $\alpha$ and $\beta$ in terms of the expectation value $\Expect{A}$ and standard deviation $\Delta A$ of $A$ in the state $\Ket{\psi}$.}

\vspace{1ex}
Although some students can do this straight away, most need some help.  During the course of the activity, I walk around the class to get an idea of how they are doing.  When it seems like many students are stuck, I reveal the following three hints in sequence.

\vspace{1ex}
\noindent{\textbf{Hints}}
\begin{enumerate}
\item \emph{Try taking the inner product of $A\Ket{\psi} = \alpha \Ket{\psi} + \beta \Ket{\psi^{\perp}}$ with other states.}

\item \emph{Try taking the inner product of $A\Ket{\psi}$ with $\Ket{\psi}$.}

\item \emph{Try taking the inner product of $A\Ket{\psi}$ with itself.}
\end{enumerate}

Although most students can get $\alpha = \Expect{A}$ either straight away or after the first hint, $|\beta| = \Delta A$ is more challenging.  After taking the inner product with $\Ket{\psi}$, the obvious instinct is to take the inner product with $\Ket{\psi^{\perp}}$, which does not help, so the third hint is usually needed.  After this, it is a short hop to the Robertson relation via the proof given in \cref{Robertson}.

I think it would be more difficult to teach the standard proof in this way.  One would either have to ask the students to derive the Cauchy-Schwarz inequality for themselves or derive the Robertson relation from Cauchy-Schwarz.  The former is a bit abstract for a quantum mechanics class and the latter involves a lot of algebra and cancellations with a high potential for making mistakes.  Both would require a large number of hints.  In contrast, the proof of the Aharonov-Vaidman identity is relatively short, and I think that students who retain the identity are more likely to be able to reconstruct the proof of the Robertson relation for themselves. 

\section{Other Uncertainty Relations for Standard Deviations}

\label{Other}

Despite the ubiquity of the Schr{\"o}dinger-Robertson uncertainty relations in quantum mechanics classes, there are good reasons to go beyond them.  For example, consider a spin-$1/2$ particle with spin operators $S_x$, $S_y$ and $S_z$.  For this case, the Robertson uncertainty is $\Delta S_x \Delta S_y \geq \hbar \Abs{\Expect{S_z}}$.  Let $\Ket{x+}$ be the spin-up state in the $x$ direction.  For this state we have $\Expect{S_z} = 0$, which is perfectly valid because $\Ket{x+}$ is an eigenstate of $S_x$ and hence $\Delta S_x = 0$.  However, because $[S_x,S_y] \neq 0$ there is necessarily some uncertainty in $S_y$ and in fact $\Delta S_y = \hbar / 2$.  The Schr{\"o}dinger relation also yields $\Delta S_x \Delta S_y \geq 0$.  So the Schr{\"o}dinger-Robertson relations do not capture all uncertainty trade-offs that necessarily exist in quantum mechanics.  

More generally, for bounded operators $A$ and $B$, any uncertainty relation of the form $\Delta A \Delta B \geq f \left ( A,B,\Ket{\psi} \right )$ for some function $f$ must necessarily have $f \left ( A,B,\Ket{\psi} \right ) = 0$ whenever $\Ket{\psi}$ is an eigenstate of $A$ or $B$.  For this reason, it makes sense to seek uncertainty relations that bound the sum of standard deviations $\Delta A + \Delta B$, the sum of variances $\left ( \Delta A\right )^2 + \left ( \Delta B \right )^2$, or more exotic combinations.  We shall discuss the Maccone-Pati relations, and some simple generalizations, in this section.

Uncertainty relations are classified as either \emph{state dependent} or \emph{state independent}, depending on whether the right hand side of the inequality depends on the state $\Ket{\psi}$.  For two observables $A$ and $B$, a state dependent uncertainty relation is of the form $f(\Delta A,\Delta B) \geq g(A, B, \Ket{\psi})$, where $f$ and $g$ are specified functions, whereas a state independent uncertainty relation would be of the form $f(\Delta A,\Delta B) \geq g(A, B)$, noting that $g$ is no longer allowed to depend on $\Ket{\psi}$.

On the face of it, a state dependent uncertainty relation is a strange idea, since, for any given normalized state $\Ket{\psi}$, we can always just calculate the uncertainties $\Delta A$ and $\Delta B$ and get the exact value of $f(\Delta A, \Delta B)$.  Therefore, bounds on uncertainty that apply to all states seem more useful.

However, a state dependent uncertainty relation can be a useful step in deriving a state independent one.  This can happen in two ways.  First, it may happen that, for a particular choice of the observables $A$ and $B$,  the function $g(A, B, \Ket{\psi})$ turns out not to depend on $\Ket{\psi}$.  For example, the Robertson relation  $\Delta A \Delta B \geq \frac{1}{2} \Abs{\Sand{\psi}{[A,B]}{\psi}}$ is state dependent, but if we choose $A = x$, $B=p$, then $\Abs{\Sand{\psi}{[A,B]}{\psi}} = 1$ and so we get the Heisenberg relation $\Delta x \Delta p \geq \frac{\hbar}{2}$, which is state independent.  Since the main point of proving the Robertson uncertainty relation in a quantum mechanics class is to give a rigorous derivation of the Heisenberg relation, its state dependence does no harm.  However, the utility of the Robertson relation for other classes of observable, such as spin components, is more questionable.  Despite the fact that I have asked students to compute it for states of a spin-$1/2$ particle as a homework problem, I do not think there is ever a need to do this in practice, as it is just as easy to calculate the exact uncertainties.

The second way of obtaining a state independent uncertainty relation from a state dependent one is to optimize, i.e.\ if $f(\Delta A,\Delta B) \geq g(A, B, \Ket{\psi})$ then\footnote{The minimum in \cref{eq:uncmin} may have to be replaced by an infimum, depending on the Hilbert space that the observables are defined on.}
\begin{equation}
\label{eq:uncmin}
f(\Delta A,\Delta B) \geq \min_{\Ket{\psi}} g(A, B, \Ket{\psi}).
\end{equation}

Of course, if $f(\Delta A,\Delta B) = \Delta A \Delta B$ and $A$ and $B$ are bounded operators then this leads to the trivial relation $\Delta A \Delta B \geq 0$ because we can choose $\Ket{\psi}$ to be an eigenstate of either $A$ or $B$.  However, for sums and more general combinations of observables, optimization can lead to a nontrivial relation.

Further, if we are considering a set of experiments that can only prepare a subset of the possible states, then we can get an uncertainty relation that applies to those states by optimizing over the subset.  An example might be experiments in which we can only prepare the system in a Gaussian state.  Although this does not yield a state independent uncertainty relation, it is more useful than a completely state dependent one, as it allows us to bound the possible uncertainties for a class of relevant states.

To summarize, state dependent uncertainty relations are a strange idea, and I am not sure whether they would ever have been considered had not Robertson introduced one as a way-point in proving the Heisenberg relation.  However, they can be useful in proving more generally applicable uncertainty relations.  The  relations that we discuss here are state dependent.

The remainder of this section is structured as follows.  In \cref{Other:Sum} we prove two propositions called the \emph{sum relations} that will be used repeatedly using the Aharonov-Vaidman identity.  In \cref{Other:SD}, we give an Aharonov-Vaidman based proof of the Maccone-Pati uncertainty relations, and in in \cref{Other:Gen} we give some simple generalizations.

\subsection{The Sum Relations}

\label{Other:Sum}

\begin{proposition}
	\label{prop:PrimSum}
	Let $A$ and $B$ be linear operators acting on \Hilb.  Then, for any $\Ket{\psi} \in \Hilb$,
	\[\Delta (A+B) \Ket{\psi^{\perp}_{A+B}} = \Delta A \Ket{\psi^{\perp}_A} + \Delta B \Ket{\psi^{\perp}_B}.\]
\end{proposition}
\begin{proof}
	Apply the Aharonov-Vaidman identity to $A+B$ in two different ways.  The first way is
	\begin{align}
		(A + B) \Ket{\psi} & = \Expect{A + B} \Ket{\psi} + \Delta (A + B) \Ket{\psi^{\perp}_{A+B}} \nonumber \\
		& = \left ( \Expect{A} + \Expect{B} \right ) \Ket{\psi} + \Delta (A + B) \Ket{\psi^{\perp}_{A+B}}, \label{eq:SumWay1}
	\end{align}
	and the second is
	\begin{align}
		(A + B) \Ket{\psi} & = A\Ket{\psi} + B\Ket{\psi} \nonumber \\ 
		& = \left ( \Expect{A} + \Expect{B} \right ) \Ket{\psi} + \Delta A \Ket{\psi^{\perp_A}} + \Delta B \Ket{\psi^{\perp}_B}. \label{eq:SumWay2}
	\end{align}
	Subtracting \cref{eq:SumWay2} from \cref{eq:SumWay1} and rearranging gives the desired result. 
\end{proof}

The next proposition comes from \cite{Pati2007}.  Here, the proof relies on \cref{prop:PrimSum} and so is based on the Aharonov-Vaidman relation.  The original proof uses a different method and is a little more complicated.

\begin{proposition}[The Sum Relation]
	\label{prop:Sum}
	Let $A$ and $B$ be two linear operators acting on a Hilbert space $\Hilb$.  Then,
	\[\Delta (A + B) \leq \Delta A + \Delta B.\]
\end{proposition}
\begin{proof}
	Let $\Ket{\psi}$ in \cref{prop:PrimSum} be a unit vector.  Then, starting from $\Delta (A+B) \Ket{\psi^{\perp}_{A+B}} = \Delta A \Ket{\psi^{\perp}_A} + \Delta B \Ket{\psi^{\perp}_B}$ and taking the inner product with $\Ket{\psi^{\perp}_{A+B}}$ gives
	\[\Delta (A + B) = \Delta A \BraKet{\psi^{\perp}_{A+B}}{\psi^{\perp}_A} + \Delta B \BraKet{\psi^{\perp_{A+B}}}{\psi^{\perp}_B}.\]
	The left hand side of this equation is a real number, so the right hand side must be too.  Therefore, we can take the real part of each term to give
	\[\Delta (A + B) = \Delta A \text{Re} \left ( \BraKet{\psi^{\perp}_{A+B}}{\psi^{\perp}_A} \right ) + \Delta B \text{Re} 
	\left ( \BraKet{\psi^{\perp_{A+B}}}{\psi^{\perp}_B} \right ),\]
	but the real part of an inner product between two unit vectors is $\leq 1$, so we have
	\[\Delta (A + B) \leq \Delta A + \Delta B.\]	
\end{proof}

From the proof, we see that the equality condition for the sum relation is
\[\text{Rcorr}(A + B, A) = \text{Rcorr}(A+B,B) = 1.\]

\begin{remark}
	For a set of linear operators $A_1,A_2,\cdots,A_n$ on a Hilbert space $\Hilb$, \Cref{prop:PrimSum} is easily generalized to
	\begin{equation}
		\label{eq:PrimSum}
		\Delta \left ( \sum_{j=1}^n A_j \right ) \Ket{\psi^{\perp}_{\sum_{j=1}^n A_j}} = \sum_{j=1}^n \Delta A_j \Ket{\psi_{A_j}^{\perp}},
	\end{equation}
	by applying the Aharonov-Vaidman identity to $\sum_{j=1}^n A_j$.  Similarly, \cref{prop:Sum} is easily generalized to
	\begin{equation}
		\label{eq:Sum}
		\Delta \left ( \sum_{j=1}^n A_j \right ) \leq \sum_{j=1}^n \Delta A_j.
	\end{equation}
	by taking the inner product of \cref{eq:PrimSum} with $\Ket{\psi^{\perp}_{\sum_{j=1}^n A_j}}$.  We will also refer to the generalization in \cref{eq:Sum} as the sum relation.
\end{remark}

\subsection{The Maccone-Pati Uncertainty Relations}

\label{Other:SD}

Between the time of Robertson's uncertainty relation and now, there has always been some literature on uncertainty relations for variances and standard deviations.  However, the field was reinvigorated in 2014, when Maccone and Pati \cite{Maccone2015} proved a pair of uncertainty relations for sums of variances, which always give a nontrivial bound, even in the case of an eigenstate of an observable.

Here, we give Aharonov-Vaidman based proofs of the Maccone-Pati relations\footnote{Although the Aharonov-Vaidman identity is used in \cite{Maccone2015}, it is not used in the proofs of the uncertainty relations.}.  

\begin{theorem}[The First Maccone-Pati Uncertainty Relation]
	\label{thm:MP1}
	Let $A$ and $B$ be Hermitian operators on a Hilbert space $\Hilb$ and let $\Ket{\psi} \in \Hilb$ be a unit vector.  Then,
	\begin{equation}
		\label{eq:MP1}
		\left ( \Delta A \right  )^2 + \left ( \Delta B \right  )^2 \geq \pm i \Expect{[A,B]} + \Abs{\Sand{\psi^{\perp}}{(A \mp i B)}{\psi}}^2,
	\end{equation}
	where $\Ket{\psi^\perp}$ is any unit vector orthogonal to $\Ket{\psi}$.
\end{theorem}
\begin{proof}
	We will prove $\left ( \Delta A \right  )^2 + \left ( \Delta B \right  )^2 \geq -i \Expect{[A,B]} + \Abs{\Sand{\psi^{\perp}}{(A + i B)}{\psi}}^2$ by applying the Aharonov-Vaidman identity to $(A + iB)$.  The proof of the other inequality follows by replacing $A + iB$ with $A - iB$.  Note that, even though $A$ and $B$ are Hermitian, $A + iB$ is not, so it is crucial that we previously generalized the Aharonov-Vaidman identity to arbitrary linear operators.
	
	Applying the Aharonov-Vaidman identity to $A+iB$ gives
	\[(A + iB) \Ket{\psi} = \left ( \Expect{A} + i \Expect{B}\right ) \Ket{\psi} + \Delta (A + iB) \Ket{\psi^{\perp}_{A+iB}}.\]
	Taking the inner product with any unit vector $\Ket{\psi^{\perp}}$ orthogonal to $\Ket{\psi}$ gives
	\[\Sand{\psi^{\perp}}{(A+iB)}{\psi} = \Delta (A + iB) \BraKet{\psi^{\perp}}{\psi^{\perp}_{A + iB}},\]
	and taking the modulus squared of this gives
	\[\Abs{\Sand{\psi^{\perp}}{(A+iB)}{\psi}}^2 = \left ( \Delta (A + iB) \right )^2 \QProb{\psi^{\perp}}{\psi^{\perp}_{A + iB}}.\]
	Now, $\Abs{\BraKet{\psi^{\perp}}{\psi^{\perp}_{A + iB}}} \leq 1$, so
	\[\left ( \Delta (A + iB) \right )^2 \geq \Abs{\Sand{\psi^{\perp}}{(A+iB)}{\psi}}^2.\]
	The result now follows by expanding $\left ( \Delta (A + iB) \right )^2$ as follows.
	\begin{align*}
		\left ( \Delta (A + iB) \right )^2 & = \Expect{(A - iB)(A + iB)} - \Expect{A - iB}\Expect{A + iB} \\
		& = \Expect{A^2} + \Expect{B^2} + i\Expect{[A,B]} - \Expect{A}^2 - \Expect{B}^2 \\
		& = \left ( \Delta A \right )^2 + \left ( \Delta B \right )^2 + i \Expect{[A,B]}.
	\end{align*}
\end{proof}

\begin{theorem}[The Second Maccone-Pati Uncertainty Relation]
	\label{thm:MP2}
	Let $A$ and $B$ be linear operators on a Hilbert space $\Hilb$ and let $\Ket{\psi} \in \Hilb$ be a unit vector.  Then,
	\begin{equation}
		\label{eq:MP2}
		\left ( \Delta A \right  )^2 + \left ( \Delta B \right  )^2 \geq \frac{1}{2} \Abs{\Sand{\psi^{\perp}_{A+B}}{(A+B)}{\psi}}^2.
	\end{equation}
\end{theorem}
\begin{proof}
	Applying the Aharonov-Vaidman identity to $A+B$ gives
	\[(A+B) \Ket{\psi} = \left ( \Expect{A} + \Expect{B} \right ) \Ket{\psi} + \Delta (A+B) \Ket{\psi^{\perp}_{A+B}}.\]
	Taking the inner product with $\Ket{\psi^{\perp}_{A+B}}$ gives
	\begin{align*}
		\Sand{\psi^{\perp}_{A+B}}{(A+B)}{\psi} & = \Delta(A + B) \\
		& \leq \Delta A + \Delta B,
	\end{align*}
	where the second line follows from the sum relation.
	
	We could stop here and regard $\Delta A + \Delta B \geq \Sand{\psi^{\perp}_{A+B}}{(A+B)}{\psi}$ as an uncertainty relation, but Maccone and Pati wanted a relation in terms of variances to compare to their first result.  To do this, we take the modulus squared of  both sides to obtain
	\[\left ( \Delta A + \Delta B\right )^2 \geq \Abs{\Sand{\psi^{\perp}_{A+B}}{(A + B)}{\psi}}^2.\]
	The result now follows from the real number inequality $x^2 + y^2 \geq \frac{1}{2}(x+y)^2$ with $x = \Delta A$ and $y = \Delta B$.  For completeness, this inequality is proved as follows.
	\begin{align*}
		&  0 \leq (x-y)^2 = x^2 + y^2 -2xy \\
		\Rightarrow \quad & x^2 + y^2 \geq 2xy \\
		\Rightarrow \quad & 2x^2 + 2y^2 \geq x^2 + y^2 + 2xy \\
		\Rightarrow \quad & 2x^2 +2 y^2 \geq (x+y)^2 \\
		\Rightarrow \quad & x^2 + y^2 \geq \frac{1}{2} (x+y)^2.
	\end{align*}
\end{proof}

\subsection{Generalizations}

\label{Other:Gen}

Generalizations of the Maccone-Pati Uncertainty relations can be obtained by applying the Aharonov-Vaidman identity to more general linear combinations $\alpha A + \beta B$, where $\alpha, \beta \in \mathbb{C}$.  This gives
\begin{equation}
	\label{eq:MPGenerator}
	(\alpha A + \beta B) \Ket{\psi} = \left ( \alpha \Expect{A} + \beta \Expect{B} \right ) \Ket{\psi} + \Delta (\alpha A + \beta B) \Ket{\psi^{\perp}_{\alpha A + \beta B}}.
\end{equation}
Applying the strategy we used to prove \cref{thm:MP1}, we can take the inner product of this with an arbitrary unit vector $\Ket{\psi^{\perp}}$ that is orthogonal to $\Ket{\psi}$, which gives
\[\Sand{\psi^{\perp}}{(\alpha A + \beta B)}{\psi} =  \Delta (\alpha A + \beta B) \BraKet{\psi^{\perp}}{\psi^{\perp}_{\alpha A + \beta B}}.\]
We can now take the modulus squared of this and recognize that  $0 \leq \QProb{\psi^{\perp}}{\psi^{\perp}_{\alpha A + \beta B}} \leq 1$ to obtain
\[\Abs{\Sand{\psi^{\perp}}{(\alpha A + \beta B)}{\psi}}^2 \leq \Delta (\alpha A + \beta B).\]
Next, we can expand $\Delta (\alpha A + \beta B)$ and rearrange to obtain
\begin{multline}
	\label{eq:GenMP1}
	\Abs{\alpha}^2 \left ( \Delta A \right )^2 + \Abs{\beta}^2 \left ( \Delta B \right )^2 \geq - \text{Re}(\alpha^* \beta) \left ( \Expect{\{A,B\}} -2 \Expect{A} \Expect{B} \right ) - i \text{Im} \left ( \alpha^* \beta \right ) \Expect{[A,B]} \\ + \Abs{\Sand{\psi^{\perp}}{(\alpha A + \beta B)}{\psi}}^2. 
\end{multline}

Substituting $\alpha = 1$, $\beta = i$ and $\alpha = 1$, $\beta = -i$ immediately yields the first Maccone-Pati Uncertainty Relation. 

Alternatively, we can apply the strategy used to prove \cref{thm:MP2}.  Starting from \cref{eq:MPGenerator}, we can take the inner product with $\Ket{\psi^{\perp}_{\alpha A + \beta B}}$ and rearrange to obtain
\[\Delta (\alpha A + \beta B) = \Sand{\psi^{\perp}_{\alpha A + \beta B}}{(\alpha A + \beta B)}{\psi}.\]
Using the sum relation, together with $\Delta (\alpha A) = |\alpha| \Delta A$ gives
\[|\alpha|\Delta A + |\beta| \Delta B \geq \Sand{\psi^{\perp}_{\alpha A + \beta B}}{(\alpha A + \beta B)}{\psi}.\]
Finally, squaring and using the inequality $x^2 + y^2 \geq \frac{1}{2}(x+y)^2$ gives
\begin{equation}
	\label{eq:GenMP2}
	|\alpha|^2 \left ( \Delta A \right )^2 + |\beta|^2 \left ( \Delta B \right )^2 \geq \frac{1}{2} \Abs{\Sand{\psi^{\perp}_{\alpha A + \beta B}}{(\alpha A + \beta B)}{\psi}}^2.
\end{equation}

The inequalities \cref{eq:GenMP1} and \cref{eq:GenMP2} are related to some of the generalizations of the Maccone-Pati uncertainty relations that have previously appeared in the literature  \cite{Bannur2015, Xiao2016}.  For example, \cref{eq:GenMP1} can be used to derive an uncertainty relation that has appeared in the literature under the name ``weighted uncertainty relation'' \cite{Xiao2016}.  To do so, we set $\alpha = \sqrt{\lambda}$, $\beta = \pm  i/\sqrt{\lambda}$ in \cref{eq:GenMP1}, where $\lambda > 0$.  This yields
\[
	\lambda \left ( \Delta A \right )^2 + \frac{1}{\lambda} \left ( \Delta B \right )^2 \geq  \pm i  \Expect{[A,B]}  + \frac{1}{\lambda} \Abs{\Sand{\psi^{\perp}}{( \lambda A \mp i B)}{\psi}}^2. 
\]
This is an uncertainty relation in its own right, but the relation in \cite{Xiao2016} comes from adding this to \cref{eq:MP1}, which yields
\[
	(1 + \lambda) \left ( \Delta A \right )^2 +\left ( 1 + \frac{1}{\lambda} \right ) \left ( \Delta B \right )^2 \geq  \pm 2 i  \Expect{[A,B]}   \Abs{\Sand{\psi^{\perp}_1}{( A \mp i B)}{\psi}}^2 + \frac{1}{\lambda} \Abs{\Sand{\psi^{\perp}_2}{( \lambda A \mp i B)}{\psi}}^2,
\]
where $\Ket{\psi_1^{\perp}}$ and $\Ket{\psi_2^{\perp}}$ are (possibly different) unit vectors that are orthogonal to $\Ket{\psi}$.

This is intended as a simple example of a generalization that is easily obtained from the Aharonov-Vaidman identity, but I expect many other uncertainty relations that are usually proved using the Cauchy-Schwarz inequality or the parallelogram law would also have simple Aharonov-Vaidman based proofs. 

\section{Quantum Propagation of Uncertainty}

\label{Prop}

In this section, we develop generalizations of the classical formulas for the propagation of uncertainty.  We start with the case of linear functions in \cref{Prop:Lin}, for which exact formulas are easy to obtain, before moving on to the general, possibly nonlinear, case  in \cref{Prop:Non}, for which we have to employ a Taylor series approximation. 

\subsection{Linear Functions}

\label{Prop:Lin}

We start with the simplest case: a sum of two observables.  Classically, if $A$ and $B$ are random variables then 
\begin{equation}
	\label{eq:classprop}
   \left [ \Delta(A + B) \right ]^2  = \left ( \Delta A \right )^2 + \left ( \Delta B \right )^2 + 2 \Delta A \Delta B \, \mathrm{corr}_{A,B}.
\end{equation}
Consider an experiment consisting of multiple runs.  On each run, the quantities $A$ and $B$ are measured.  These quantities are formalized as random variables because we assume that our experiments are subject to random statistical fluctuations, and that the ``true'' values that we are seeking are the means $\Expect{A}$ and $\Expect{B}$ of these random processes.  We then use the average values calculated from the data as estimates of $\Expect{A}$ and $\Expect{B}$, and the standard deviations as a measure of the error in our experiment.    If we are actually interested in the quantity $A+B$ then we would sum the averages to form our estimate of $\Expect{A+B}$, and we would use \cref{eq:classprop} to determine the error in our estimate of $\Expect{A+B}$.  Using \cref{eq:classprop} in this way is called the \emph{propagation of uncertainty} or \emph{propagation of error}.

If the random variables, $A$ and $B$ are independent, which would be the case if the randomness were due to independent statistical errors, then $\mathrm{corr}_{A,B} = 0$ and we would have
\[ \left [ \Delta(A + B) \right ]^2  = \left ( \Delta A \right )^2 + \left ( \Delta B \right )^2, \]
which is the formula for propagation of uncertainty that is most commonly used in practice.

We now want to generalize these formulas by replacing classical random variables with quantum observables.  The generalization of \cref{eq:classprop} is as follows.
\begin{theorem}
	\label{thm:LinProp}	
	Let $A$ and $B$ be Hermitian operators on a Hilbert space $\Hilb$.  Then,
	\begin{align}
		\label{eq:LinProp}
		\left [ \Delta(A + B) \right ]^2 & = \left ( \Delta A \right )^2 + \left ( \Delta B \right )^2 + 2 \Delta A \Delta B \, \mathrm{Rcorr}_{A,B} \\
		& = \left ( \Delta A \right )^2 + \left ( \Delta B \right )^2 + \Expect{\{A,B\}} - 2\Expect{A}\Expect{B}
	\end{align}
\end{theorem}
\begin{proof}
	\Cref{prop:PrimSum} implies that, for any unit vector $\Ket{\psi} \in \Hilb$,
	\[\Delta (A + B) \Ket{\psi_{A+B}^{\perp}} = \Delta A \Ket{\psi_A^{\perp}} + \Delta B \Ket{\psi_B^{\perp}}.\]
	Taking the inner product of this with itself gives
	\begin{align*}
		\left [ \Delta \left ( A + B \right )\right ]^2 & = \left ( \Delta A \right )^2 + \left ( \Delta B \right )^2 + \Delta A \Delta B \left ( \BraKet{\psi_A^{\perp}}{\psi_B^{\perp}} + \BraKet{\psi_B^{\perp}}{\psi_A^{\perp}} \right ) \\
		& = \left ( \Delta A \right )^2 + \left ( \Delta B \right )^2 + 2\Delta A \Delta B \, \mathrm{Re}\left ( \BraKet{\psi_A^{\perp}}{\psi_B^{\perp}} \right ).
	\end{align*}
	Applying \cref{eq:Rcorr} completes the proof.
\end{proof}

\begin{remark}
	For operators $A_1,A_2,\cdots,A_n$ and real numbers $\alpha_1,\alpha_2,\cdots,\alpha_n$, \cref{thm:LinProp} is easily generalized to
	\begin{align*}
		\left [ \Delta \left ( \sum_{j=1}^n \alpha_j A_j \right ) \right ]^2 & = \sum_{j=1}^n \alpha_j^2 \left ( \Delta A_j \right )^2 + \sum_{j \neq k} \alpha_j \alpha_k \Delta A_j \Delta A_k \, \mathrm{Rcorr}_{A_j,A_k} \\ 
		& = \sum_{j=1}^n \alpha_j^2 \left ( \Delta A_j \right )^2 + \sum_{j \neq k} \alpha_j \alpha_k \left ( \Expect{\{ A_j,A_k \}} - 2 \Expect{A_j} \Expect{A_k} \right ). 
	\end{align*}
\end{remark}

Although \cref{thm:LinProp} is a true theorem about quantum observables, it cannot be used to propagate uncertainty in the same way as its classical counterpart.  Classically, we can measure $A$ and $B$ together in the same run of the experiment.  We can then estimate $A+B$ by summing the average values of $A$ and $B$ that we found in the experiment.  We also have all the information we need to calculate the uncertainty $\Delta (A+B)$, i.e.\ $\Delta A$, $\Delta B$, $\Expect{A}$, $\Expect{B}$ and $\Expect{AB}$, so we can determine the uncertainty without doing any more experiments.  

In quantum mechanics, this is not the case.  When $A$ and $B$ do not commute, they cannot both be accurately measured on the same run of an experiment.  We can still estimate their expectation values by measuring $A$ on half of the runs of the experiment and $B$ on the other half and taking averages.  Since $\Expect{A + B} = \Expect{A} + \Expect{B}$, summing these averages is still a way of estimating $\Expect{A+B}$.  However, we do not have enough information to calculate $\Delta (A+B)$.  The reason is that $\Delta (A+B)$ is the uncertainty in a \emph{direct} measurement of $A + B$.  Since $A$ and $B$ do not commute, this requires a different experimental setup from a measurement of $A$ and $B$ alone.

If we wanted to use \cref{eq:LinProp} to calculate $\Delta (A+B)$, we would also have to estimate $\Expect{\{ A,B\}}$.  The most straightforward way of doing this would be to measure the observable $\{A,B\} = AB + BA$, but this requires yet another different experimental setup, and one that is likely to be at least as complicated as measuring $A+B$ directly.

An exception to this are cases where $\{A,B\} = c I$ for some constant $c$, in which case $\Expect{\{ A,B\}} = c$ regardless of the state.  In particular, this is true of the Pauli observables $\sigma_x$, $\sigma_y$, $\sigma_z$ of a qubit for which $\{\sigma_j,\sigma_k\} = \delta_{jk} I$, where $j$ and $k$ run over $x,y,z$.  Therefore, if we measure $\sigma_x$ on many qubits prepared in the same way and $\sigma_y$ on another set of such qubits, we can estimate $\Expect{\sigma_x + \sigma_y}$ and $\Delta (\sigma_x + \sigma_y)$ without doing any further experiments using the formula
\[\left [ \Delta \left ( \sigma_x + \sigma_y \right ) \right ]^2= \left ( \Delta \sigma_x \right )^2 + \left ( \Delta \sigma_y \right )^2  - 2\Expect{\sigma_x}\Expect{\sigma_y}.\]

When $\{A,B\} \neq c I$, I do not know of any situations in which \cref{eq:LinProp} would be useful in practice, but from a theoretical point of view it is the appropriate generalization of \cref{eq:classprop} to quantum mechanics, and this bolsters the case that $\text{Rcorr}_{A,B}$ is the appropriate quantum generalization of the classical correlation. 

\subsection{Nonlinear Functions}

\label{Prop:Non}

For nonlinear functions $f(A,B)$ of two random variables $A$ and $B$, it is common to use a first order Taylor expansion of $f(A,B)$ about the point $f(\Expect{A},\Expect{B})$ to derive an approximation for the variance $[\Delta f(A,B)]^2$ to second order in $\Delta A$ and $\Delta B$.  This yields the formula
\begin{align*}
	\left [ \Delta f(A,B) \right ]^2 & \approx \left ( \left . \frac{\partial f}{\partial A} \right \vert_{A = \Expect{A},B=\Expect{B}} \right )^2 \left ( \Delta A \right )^2  + \left ( \left . \frac{\partial f}{\partial B} \right \vert_{A=\Expect{A}, B = \Expect{B}} \right )^2  \left ( \Delta B \right )^2 \\ & \qquad\qquad\qquad + \left . \frac{\partial f}{\partial A} \right \vert_{A = \Expect{A},B=\Expect{B}} \left . \frac{\partial f}{\partial B} \right \vert_{A=\Expect{A},B = \Expect{B}} \Delta A \Delta B \, \mathrm{corr}_{A,B}.
\end{align*}
To avoid cluttering notation, I will write $\bar{A}$ for $A= \Expect{A}$, so that we can more compactly write
\begin{equation}
	\label{eq:genclassprop}
	\left [ \Delta f(A,B) \right ]^2 \approx \left . \frac{\partial f}{\partial A} \right \vert_{\bar{A},\bar{B}}^2 \left ( \Delta A \right )^2  + \left . \frac{\partial f}{\partial B} \right \vert_{\bar{A}, \bar{B}}^2  \left ( \Delta B \right )^2  + \left . \frac{\partial f}{\partial A} \right \vert_{\bar{A},\bar{B}} \left . \frac{\partial f}{\partial B} \right \vert_{\bar{A},\bar{B}} \Delta A \Delta B \, \mathrm{corr}_{A,B}. 
\end{equation}
When $A$ and $B$ are independent, this reduces to
\[ \left [ \Delta f(A,B) \right ]^2  \approx  \left . \frac{\partial f}{\partial A} \right \vert_{\bar{A},\bar{B}}^2 \left ( \Delta A \right )^2  + \left . \frac{\partial f}{\partial B} \right \vert_{\bar{A},\bar{B}}^2  \left ( \Delta B \right )^2,\]
which is the most commonly used form.

The quantum generalization of \cref{eq:genclassprop} is as follows.

\begin{theorem}
	\label{thm:GenProp}
	Let $A$ and $B$ be Hermitian operators on a Hilbert space $\Hilb$ and consider a function $f: \mathfrak{H}(\Hilb) \times \mathfrak{H}(\Hilb) \rightarrow \mathfrak{H}(\Hilb)$ where $\mathfrak{H}(\Hilb)$ is the space of Hermitian operators on $\Hilb$.  Then
	\begin{equation}
	 	\label{eq:GenProp}
		 \left [ \Delta f(A,B) \right ]^2  \approx  \left . \frac{\partial f}{\partial A} \right \vert_{\bar{A},\bar{B}}^2 \left ( \Delta A \right )^2  + \left . \frac{\partial f}{\partial B} \right \vert_{\bar{A}, \bar{B}}^2  \left ( \Delta B \right )^2  + \left . \frac{\partial f}{\partial A} \right \vert_{\bar{A},\bar{B}} \left . \frac{\partial f}{\partial B} \right \vert_{\bar{A},\bar{B}} \Delta A \Delta B \, \mathrm{Rcorr}_{A,B}
	\end{equation}
	where $\approx$ means equality to second order in $\Delta A$ and $\Delta B$
\end{theorem}
\begin{proof}
	Consider the first order Taylor expansion of $f(A,B)$ about the point $f_0 = f(\Expect{A},\Expect{B})$,
	\[f(A,B) \approx f_0 + \left . \frac{\partial f}{\partial A} \right \vert_{\bar{A},\bar{B}} A + \left . \frac{\partial f}{\partial B} \right \vert_{\bar{A},\bar{B}}  B.\]
	Applying \cref{prop:PrimSum} to this gives
	\[\left [ \Delta f(A,B) \right ] \Ket{\psi^{\perp}_{f(A,B)}} \approx \left . \frac{\partial f}{\partial A} \right \vert_{\bar{A},\bar{B}} \Delta A \Ket{\psi^{\perp}_A} + \left . \frac{\partial f}{\partial B} \right \vert_{\bar{A},\bar{B} } \Delta B \Ket{\psi^{\perp}_B}.\]
	Taking the inner product of this with itself gives
	\begin{align*}
		\left [ \Delta f(A,B) \right ]^2 & \approx  \left . \frac{\partial f}{\partial A} \right \vert_{\bar{A},\bar{B}}^2 \left ( \Delta A \right )^2 +  \left . \frac{\partial f}{\partial B} \right \vert_{\bar{A}, \bar{B}}^2 \left ( \Delta B \right )^2 + \left . \frac{\partial f}{\partial A} \right \vert_{\bar{A},\bar{B}} \left . \frac{\partial f}{\partial B} \right \vert_{\bar{A}, \bar{B}} \Delta A \Delta B \, \mathrm{Re} \left ( \BraKet{\psi_A^{\perp}}{\psi_B^{\perp}} \right ) \\
& = \left . \frac{\partial f}{\partial A} \right \vert_{\bar{A},\bar{B}}^2 \left ( \Delta A \right )^2 +  \left . \frac{\partial f}{\partial B} \right \vert_{\bar{A}, \bar{B}}^2 \left ( \Delta B \right )^2 + \left . \frac{\partial f}{\partial A} \right \vert_{\bar{A},\bar{B}} \left . \frac{\partial f}{\partial B} \right \vert_{\bar{A}, \bar{B}} \Delta A \Delta B \, \mathrm{Rcorr}_{A,B}	 
	\end{align*}
\end{proof}

\begin{remark}
	For operators $A_1,A_2,\cdots,A_n$ and a function $f(A_1,A_2,\cdots,A_n)$, \cref{thm:GenProp} is easily generalized to
	\[
		\left [ \Delta f \left ( A_1,A_2,\cdots ,A_n \right ) \right ]^2 \approx \sum_{j=1}^n \left. \frac{\partial f}{\partial A_j} \right \vert_{\bar{A}}^2 \left ( \Delta A_j \right )^2  + \sum_{j \neq k} \left. \frac{\partial f}{\partial A_j} \right \vert_{\bar{A}} \left. \frac{\partial f}{\partial A_k} \right \vert_{\bar{A}} \Delta A_j \Delta A_k \, \mathrm{Rcorr}_{A_j,A_k},
	\]
       where $\bar{A}$ is shorthand for $A_1 = \Expect{A_1}, A_2 = \Expect{A_2},\cdots A_n = \Expect{A_n}$.
\end{remark}

As a formula for propagating uncertainty, \cref{eq:GenProp} inherits all of the problems of \cref{eq:LinProp}, but the problems are compounded further by use of the first order Taylor approximation.  This approximation is valid when $\Delta A$ and $\Delta B$ are suitably small compared to $\Expect{A}$, $\Expect{B}$, $f(\Expect{A},\Expect{B})$ and the derivatives of $f(A,B)$ at $A=\Expect{A}$, $B=\Expect{B}$.  This is often the case in classical experiments where everything can be measured with a small statistical error.  However, in quantum mechanics, when $A$ and $B$ do not commute, the (various) uncertainty relations tell us that there is necessarily a trade-off between the size of $\Delta A$ and $\Delta B$.  If one of them is small, then the other might necessarily have to be large.  For example, for the Pauli observables $\sigma_x$ and $\sigma_y$, at least one of the uncertainties must be comparable in size to $1$, which is the largest possible value of $\Expect{\sigma_x}$ or $\Expect{\sigma_y}$.

A case where the formula will work well is for a continuous variable system where $\Delta x \sim \Delta p \sim \sqrt{\hbar}$, and $\Expect{x}$, $\Expect{p}$ are large compared to $\sqrt{\hbar}$.  But this is a case where you would expect classical physics to be a good approximation anyway.  

I do not know whether there is a practical use of \cref{eq:GenProp}, but it is nonetheless a correct formal generalization of \cref{eq:genclassprop}.

\section{Dealing with Mixed States}

\label{Mixed}

So far, we have dealt exclusively with the case of pure state vectors $\Ket{\psi}$.  However, all of our results can be extended to more general density operators $\rho$, which can represent mixed states.  The most familiar way to do this is to make use of the concept of a \emph{purification} of a density operator.  Given a density operator on a Hilbert space $\Hilb[S]$, where $S$ stands for ``system'', we can always find a pure state vector $\Ket{\psi}_{SE} \in \Hilb[S]\otimes \Hilb[E]$, where $E$ is the ``environment'', such that
\[\rho_S = \Tr[E]{\Proj{\psi}_{SE}},\]
and $\mathrm{Tr}_E$ is the partial trace over $\Hilb[E]$.  You can then apply the Aharonov-Vaidman identity to operators of the form $A_S \otimes I_E$ acting on a purification to obtain results about the density operator $\rho_S$.

However, to make the parallels to the pure state case as close as possible, I prefer to use an equivalent concept, called an \emph{amplitude operator}.  The equivalence between amplitude operators and purifications is discussed in \cref{App:Purify}
\begin{definition}
\label{def:Amplitude}
Given a density operator $\rho_S$ on a Hilbert space $\Hilb[S]$, an \emph{amplitude operator} for $\rho_S$ is a linear operator $L_S : \Hilb[E] \rightarrow \Hilb[S]$, where $\Hilb[E]$ is any Hilbert space, such that
\[\rho_S = L_S L_S^{\dagger}.\]
\end{definition}

The reason for the name \emph{amplitude operator} is that, in pure-state quantum mechanics, an amplitude is a complex number $\alpha$ such that $|\alpha|^2$ is a probability.  A density operator is a non-commutative generalization of a probability distribution \cite{Redei2007, Leifer2013}, and hence an amplitude operator ought to be an operator that ``squares'' to a density operator.

Given a density operator $\rho_S$, one obvious way of constructing an amplitude operator is to set $\Hilb[E]=\Hilb[S]$ and $L_S = \sqrt{\rho}_S$, but there are an infinite number of alternatives, as the following proposition shows
\begin{proposition}
	An operator $L_S:\Hilb[E]\rightarrow\Hilb[S]$ is an amplitude operator for $\rho_S$ if and only if
	\[L_S = \sqrt{\rho}_S U_{S|E},\]
	where $U_{S|E}:\Hilb[E] \rightarrow \Hilb[S]$ is a semi-unitary operator, i.e.\ it satisfies $U_{S|E}U_{S|E}^{\dagger} = I_S$
\end{proposition}
\begin{proof}
An operator of the form $L_S = \sqrt{\rho}_SU_{S|E}$ obviously satisfies \cref{def:Amplitude}.  For the other direction, assume $L_S$ is an amplitude operator.  Like any operator, it may be decomposed in its polar decomposition $L_S = P_S U_{S|E}$ where $P_S$ is a positive semi-definite operator on $\Hilb[S]$, and $U_{S|E}: \Hilb[E] \rightarrow \Hilb[S]$ is semi-unitary\footnote{The polar decomposition is often only defined for square matrices, in which case $\Hilb[E] = \Hilb[S]$ and $U_{S|E}$ is unitary.  Here, we use the generalization to non-square matrices (see e.g. \cite{BenIsrael2003}).}.  The definition of an amplitude operator then implies that $\rho_S = P_S U_{S|E} U_{S|E}^{\dagger} P_S = P_S^2$, so we must have $P_S = \sqrt{\rho}_S$.  
\end{proof}

Going back to the analogy between amplitudes and amplitude operators, multiplying an amplitude $\alpha$ by a phase factor $e^{i\phi}$ does not change the probability it represents.  Similarly, multiplying an amplitude operator $L_S$ by a semi-unitary $V_{E|E'}$, i.e.\ an operator $V_{E|E'}:\Hilb[E']\rightarrow\Hilb[E]$ satisfying $V_{E|E'} V_{E|E'}^{\dagger} = I_E$, on the right does not change the density operator it represents.  Although one might think it desirable to work directly with probabilities or density operators in order to eliminate these ambiguities, the mathematical manipulations we need to do in quantum mechanics are often linear in the amplitudes or amplitude operators, but would be nonlinear if you used probabilities or density operators.  Therefore, it is often more convenient to live with the ambiguity.

Since \emph{every} operator has a polar decomposition, the only requirement for $L_S$ to be an amplitude operator for \emph{some} density operator is that $\Tr[S]{L_S L_S^{\dagger}} = 1$.  If we want to work with unnormalized density operators, i.e.\ any positive operator, then any operator $L_S:\Hilb[E] \rightarrow \Hilb[S]$ is the amplitude operator for some (possibly unnormalized) density operator.  This is analogous to the fact that any vector in $\Hilb[S]$ represents a (possibly unnormalized) pure state.

The strategy for generalizing the Aharonov-Vaidman identity, and everything that follows from it, is to replace the state vector $\Ket{\psi}_S$ with an amplitude operator $L_S$.    The reason this works is that the space of linear operators mapping $\Hilb[E]$ to $\Hilb[S]$, which we denote $\mathfrak{L}_{S|E}$, is itself a Hilbert space with inner product $\Expect{L_S,M_S} = \Tr[E]{L_S^{\dagger}M_S}$, known as the \emph{Hilbert-Schmidt} inner product\footnote{By the cyclic property of the trace, we can also write $\Expect{L_S,M_S} = \Tr[S]{M_S L_S^{\dagger}}$.}.  Since the Aharonov-Vaidman identity is valid for any Hilbert space, it must be valid on $\mathfrak{L}_{S|E}$ as well.
\begin{proposition}[The Aharonov-Vaidman Identity for Operators]
	\label{Prop:AVOp}
	Let $A_S$ be a linear operator on a Hilbert space $\Hilb[S]$ and let $L_S:\Hilb[E] \rightarrow \Hilb[S]$.  Then,
	\begin{equation}
		\label{eg:GenAVOp}
		A_S L_S = \Expect{A_S} L_S + \left ( \Delta A_S \right ) L_{A_S}^{\perp},
	\end{equation}
	where $\Expect{A_S} = \Tr[S]{A_SL_SL_S^{\dagger}}/\Tr[S]{L_S L_S^{\dagger}}$, $\Delta A = \sqrt{\Expect{A_S^{\dagger}A_S} - \Abs{\Expect{A_S}}^2}$, and $L^{\perp}_{A_S}:\Hilb[E]\rightarrow \Hilb[S]$ is an amplitude operator that is orthogonal to $L_S$, i.e. $\Tr[E]{L^{\dagger}_S L^{\perp}_{A_S}} = 0$, satisfies $\Tr[S]{L^{\perp}_{A_S}L_{A_S}^{\perp \dagger}} = \Tr[S]{L_S L_S^{\dagger}}$, and depends on both $L_S$ and $A_S$.
\end{proposition}

The proof of this proposition is essentially the same as the proof of the vector Aharonov-Vaidman identity (\cref{Prop:AV}) with the standard inner product replaced by the Hilbert-Schmidt inner product.  The only difference is that the cyclic property of the trace is also needs to be used to write things in the exact form given in \cref{Prop:AVOp}.  I leave this as an exercise for the reader. 

Since $\rho_S = L_S L_S^{\dagger}$ is always a (possibly unnormalized) density operator, we can write
\[
	\Expect{A_S}  = \frac{\Tr[S]{A_SL_SL_S^{\dagger}}}{\Tr[E]{L_S^{\dagger}L_S}}  = \frac{\Tr[S]{A_S \rho_S}}{\Tr[E]{\rho_S}}.
\]
We can also introduce the density operator $\rho_{A_S}^{\perp} = L^{\perp}_{A_S} L_{A_S}^{\perp \dagger}$, which will be normalized in the same way as $\rho_S$, i.e., \Tr[S]{\rho_{A_S}^{\perp}} = \Tr[S]{\rho_S}.

When $L_S$ is normalized so that $\rho_S = L_S L_S^{\dagger}$ is a normalized density operator,  i.e., $\Tr[S]{L_SL_S^{\dagger}} = 1$, then $\rho_{A_S}^{\perp}$ is also normalized, i.e., \Tr[S]{\rho_{A_S}^{\perp}} = 1.

As defined, $\rho_{A_S}^{\perp} = L^{\perp}_{A_S} L_{A_S}^{\perp \dagger}$ looks like it depends on the choice of amplitude operator $L_S$.  In fact, it does not.  It only depends on $\rho_S$ and $A_S$.  To see this, rewrite the operator Aharonov-Vaidman identity as
\[
		L_{A_S}^{\perp} = \frac{1}{\Delta A_S} \left ( A_S  - \Expect{A_S} I_S \right ) L_S,
\]
and then we have,
\begin{align*}
\rho_{A_S}^{\perp} & = L^{\perp}_{A_S} L_{A_S}^{\perp \dagger} \\
& = \frac{1}{(\Delta A_S)^2} \left ( A_S  - \Expect{A_S} I_S \right ) L_S L_S^{\dagger} \left ( A_S^{\dagger}  - \Expect{A_S}^* I_S \right ) \\
& = \frac{1}{(\Delta A_S)^2} \left ( A_S  - \Expect{A_S} I_S \right ) \rho_S \left ( A_S^{\dagger}  - \Expect{A_S}^* I_S \right ),
\end{align*}
which is clearly independent of the choice of $L_S$.  Note that, although $L_S$ and $L_{A_S}^{\perp}$ are Hilbert-Schmidt orthogonal, $\rho_S$ and $\rho_{A_S}^{\perp}$ are generally not.

To generalize the results of this paper from state vectors to density operators, we replace the vector Aharonov-Vaidman identity with its operator counterpart applied to amplitude operators, and we replace the usual inner product with the Hilbert-Schmidt inner product.  In many cases, the final result is independent of the amplitude operator used to represent the state.  Although we use it in the proof, it drops out in the final result by only appearing in the combination $L_S L_S^{\dagger}$, as in the expression we derived for $\rho_{A_S}^{\perp}$.  In fact, the final formulas are usually the same as in the pure state case, except that we have to interpret $\Expect{A_S}$ as $\Tr[S]{A_S \rho_S}$ rather than $\Sand{\psi}{A_S}{\psi}$.

However, this is not true for the Maccone-Pati uncertainty relations and their generalizations, which do depend on the choice of amplitude operator $L_S$.  

\begin{theorem}[The First Maccone-Pati Uncertainty Relation for amplitude operators]
	\label{thm:MP1Op}
	Let $A_S$ and $B_S$ be Hermitian operators on a Hilbert space $\Hilb[S]$ and let $\rho_S$ be a normalized density operator on $\Hilb[S]$.  Then,
	\begin{equation}
		\label{eq:MP1Op}
		\left ( \Delta A \right  )^2 + \left ( \Delta B \right  )^2 \geq \pm i \Expect{[A,B]} + \Abs{\Tr[E]{L_S^{\perp \dagger} (A \mp i B)L_S}}^2,
	\end{equation}
	where $L_S: \Hilb[E] \rightarrow \Hilb[S]$ is any amplitude operator for $\rho_S$, and $L_S^{\perp}: \Hilb[E] \rightarrow \Hilb[S]$ is any normalized amplitude operator orthogonal to $L_S$ that has the same input space $\Hilb[E]$.
\end{theorem}

Note that, in order to obtain the tightest possible bound on $\left ( \Delta A \right  )^2 + \left ( \Delta B \right  )^2$, the right hand side of \cref{eq:MP1Op} should be maximized over all possible choices of $L_S$ and $L_S^{\perp}$.  To do this in practice, a bound on the largest dimension $d_E$ required to obtain the maximum is needed.  I conjecture that $d_E = 2d_S$ is sufficient because this allows $L_S$ and $L_S^{\perp}$ to have orthogonal kernels on $\Hilb[E]$, but I do not have a proof of this.

\begin{theorem}[The Second Maccone-Pati Uncertainty Relation for amplitude operators]
	\label{thm:MP2Op}
	Let $A_S$ and $B_S$ be linear operators on a Hilbert space $\Hilb[S]$ and let $\rho_S$ be a normalized density operator on $\Hilb[S]$.  Then,
	\begin{equation}
		\label{eq:MP2Op}
		\left ( \Delta A_S \right  )^2 + \left ( \Delta B_S \right  )^2 \geq \frac{1}{2} \Abs{\Tr[E]{L^{\perp \dagger}_{A_S + B_S} (A+B) L_S}}^2,
	\end{equation}
      where $L_S$ is any amplitude operator for $\rho_S$ and
	\[ L_{A_S+B_S}^{\perp} = \frac{1}{\Delta (A_S + B_S)} \left ( A_S + B_S  - \Expect{A_S + B_S} I_S \right ) L_S. \]
\end{theorem}

In this case, to obtain the tightest bound, we have to maximize the right hand side over $L_S$.  We do not have to separately optimize over $L_{A_S+B_S}^{\perp}$ because it is a function of $L_S$, $A_S$ and $B_S$.  However, its dependence on $L_S$ makes the problem into a complicated nonlinear optimization.

\section{Summary and Conclusions}

\label{Conc}

In this paper, I discussed how the standard textbook uncertainty relations of Robertson and Schr{\"o}dinger can be derived from the Aharonov-Vaidman identity in a more direct way than the standard proof.  I also demonstrated the identity's usefulness in proving other uncertainty relations, such as the Maccone-Pati relations, and the quantum formulas for propagation of uncertainty.  Finally, I gave a mixed-state generalization of the Aharonov-Vaidman identity in terms of amplitude operators.  I hope that this has persuaded you that the Aharonov-Vaidman identity belongs in undergraduate textbooks and that it ought to be a first-line tool in proving relationships between standard deviations in quantum mechanics.  I am sure there are other uncertainty relations that have an elegant Aharonov-Vaidman based proofs, and I hope to find new and useful uncertainty relations that have not been discovered before via this method.

The Aharonov-Vaidman identity naturally gives rise to two quantum generalizations of the correlation, $\mathrm{corr}_{A,B}$ and $\mathrm{Rcorr}_{A,B}$.  It would be interesting to determine whether these quantities have an operational meaning in the case where $A$ and $B$ do not commute.  On the more formal side, perhaps there is a pseudo-probability representation of quantum mechanics, such as the Wigner function \cite{Wigner1932, Gross2006, Curtright2014} or the Kirkwood-Dirac distribution \cite{Kirkwood1933, Dirac1945, Lostaglio2022}, for which these are the correlations for observables as defined on the appropriate phase space.  This might help to find uses for the propagation of error formulas in cases where the observables do not commute.

\section*{Acknowledgments}

I would like to thank Yakir Aharonov for introducing me to the Aharonov-Vaidman identity and emphasizing its importance.  I would like to acknowledge (but not thank) the role played by the COVID19 pandemic shutdowns in giving me the opportunity to think about uncertainty relations and their pedagogy.  This research was supported in part by the Fetzer Franklin Fund of the John E. Fetzer Memorial Trust and by grant number FQXi-RFPIPW-1905 from the Foundational Questions Institute and Fetzer Franklin Fund, a donor advised fund of Silicon Alley Community Foundation.  This research was also supported in part by Perimeter Institute for Theoretical Physics.  Research at Perimeter Institute is supported by the Government of Canada through the Department of Innovation, Science, and Economic Development, and by the Province of Ontario through the Ministry of Colleges and Universities.

\bibliographystyle{unsrturl}
\bibliography{AharonovIdentity}

\appendix

\section{Amplitude Operators and Purifications}

\label{App:Purify}

\begin{proposition}
   Given a density operator $\rho_S$ on a Hilbert space $\Hilb[S]$, let $\Hilb[S']$ be another copy of the same Hilbert space and let $\{\Ket{j}\}$ be an orthonormal basis for $\Hilb[S]$ and $\Hilb[S']$.  Define the vector
\[\Ket{\Phi^+}_{SS'} = \sum_j \Ket{j}_S \Ket{j}_{S'}.\]

Let $L_S:\Hilb[E] \rightarrow \Hilb[S]$ be an amplitude operator for $\rho_S$ and let $\{\Ket{k}_E\}$ be an orthonormal basis for $\Hilb[E]$.  Then $I_S \otimes L_{S'}^T \Ket{\Phi^+}_{SS'}$ is a purification of $\rho_S$, where $^T$ denotes transpose in the $\Outer{j}{k}_{SE}$ basis.  Similarly, if $\Ket{\psi}_{SE} \in \Hilb[S]\otimes\Hilb[E]$ is a purification of $\rho_S$ then $L_S = \Bra{\psi^*}_{S'E}\Ket{\Phi^+}_{SS'} $ is an amplitude operator for $\rho_S$, where $^*$ denotes complex conjugation in the $\Ket{jk}_{S'E}$ basis.
\end{proposition}
\begin{proof}
If $L_S$ is an amplitude operator for $\rho_S$ then $\rho_S = L_S L_S^{\dagger}$.  We have to show that this implies that $\Tr[E]{I_S \otimes L_{S'}^T \Proj{\Phi^+}_{SS'} I_S \otimes \left ( L_{S'}^T \right )^{\dagger}} = \rho_S$.  Note that $\left ( L_{S'}^T \right )^{\dagger} = L_{S'}^*$, where $^*$ denotes complex conjugate in the $\Outer{j}{k}_{SE}$ basis.  Therefore, we have
\begin{align}
	\Tr[E]{I_S \otimes L_{S'}^T \Proj{\Phi^+}_{SS'} I_S \otimes  L_{S'}^* } & = \sum_{j,k} \Outer{j}{k}_S \Tr[E]{L_{S'}^T \Outer{j}{k}_{S'}L_{S'}^*} \\
& = \sum_{j,k} \Outer{j}{k}_S \Sand{k}{L_S^* L_S^T}{j}_S,
\end{align}
where we have changed the index $S'$ to $S$ because they refer to the same Hilbert space and $\Sand{k}{L_S^T L_S^*}{j}_S$ is a scalar.  Rearranging this, we have
\begin{align}
\Tr[E]{I_S \otimes L_{S'}^T \Proj{\Phi^+}_{SS'} I_S \otimes  L_{S'}^* } & = \sum_{j,k} \Ket{j}_S \Sand{k}{L_S^*L_S^T}{j}_S  \Bra{k}_{S} \\
& = \sum_{j,k} \Proj{j}_S \left ( L_S^* L_S^T \right )^T \Proj{k}_S \\
& = \left ( L_S^* L_S^T \right )^T = L_S L_S^{\dagger} = \rho_S.
\end{align}

For the other direction, we have to prove that $L_S L_S^{\dagger} = \rho_S$, where $L_S = \Bra{\psi^*}_{S'E}\Ket{\Phi^+}_{SS'} $ and $\Ket{\psi}_{SE}$ is any purification of $\rho_S$, i.e.\ $\Tr[E]{\Proj{\psi}_{SE}} = \rho_S$.  

First, let $\Ket{\psi}_{SE} = \sum_{jk} \alpha_{jk} \Ket{j}_S \otimes \Ket{k}_E$ be the decomposition of $\Ket{\psi}_{SE}$ in the $\Ket{jk}_{SE}$ basis.  We have $\Ket{\psi^*}_{SE} = \sum_{jk} \alpha_{jk}^* \Ket{j}_S \otimes \Ket{k}_E$ and the condition  $\Tr[E]{\Proj{\psi}_{SE}} = \rho_S$ is equivalent to $\sum_{j,k,l} \alpha_{jk}\alpha_{lk}^* \Outer{j}{l}_S = \rho_S$.  Note also that $\Bra{j}_{S'}\Ket{\Phi^+}_{SS'} = \Ket{j}_S$.

Hence, we have
\begin{align}
	L_S L_S^{\dagger} & = \left ( \Bra{\psi^*}_{S'E}\Ket{\Phi^+}_{SS'} \right ) \left ( \Bra{\psi^*}_{S'E}\Ket{\Phi^+}_{SS'} \right )^{\dagger} \\
	& =  \Bra{\psi^*}_{S'E}\Ket{\Phi^+}_{SS'} \Bra{\Phi^+}_{SS'} \Ket{\psi^*}_{S' E} \\
	& = \sum_{jklm} \alpha_{jk}  \Bra{j}_{S'}\Bra{k}_{E'} \Proj{\Phi^+}_{SS'} \alpha^{*}_{lm} \Ket{l}_{S'} \Ket{m}_E \\
	& = \sum_{jklm} \alpha_{jk}\alpha^*_{lm} \BraKet{k}{m}_E  \left ( \Bra{j}_{S'}\Ket{\Phi^+}_{SS'} \right ) \left ( \Bra{\Phi^+}_{SS'} \Ket{l}_{S'} \right ) \\
	& = \sum_{jkl} \alpha_{jk} \alpha^*_{lk} \Outer{j}{l}_{S} \\
	& = \rho_S.
\end{align}
\end{proof}

\end{document}